\documentclass[11pt, reqno]{amsart}

\usepackage{amsmath, amsfonts, amssymb, amsbsy, bigstrut, graphicx, enumerate,  upref, longtable, comment}

\usepackage{pstricks}

\usepackage[numbers, sort&compress]{natbib} 

\usepackage[breaklinks]{hyperref}

\newcommand{\ep}{\epsilon}

\newtheorem{thm}{Theorem}[section]
\newtheorem{lmm}[thm]{Lemma}

\newtheorem{prop}[thm]{Proposition}
\theoremstyle{definition}

\newcommand{\cov}{\mathrm{Cov}}
\newcommand{\dm}{\mathcal{D}}

\newcommand{\ee}{\mathbb{E}}

\newcommand{\mx}{\mathcal{X}}

\newcommand{\rr}{\mathbb{R}}
\newcommand{\smallavg}[1]{\langle #1 \rangle}

\newcommand{\st}{\sqrt{t}}
\newcommand{\sst}{\sqrt{1-t}}

\newcommand{\var}{\mathrm{Var}}

\newcommand{\xp}{X^\prime}

\newcommand{\xx}{\mathcal{X}}

\newcommand{\zz}{\mathbb{Z}}

\numberwithin{equation}{section}

\newcommand{\eq}[1]{\begin{align*} #1 \end{align*}}

\newcommand{\tv}{d_{\textup{TV}}}

\newcommand{\si}{\langle \sigma_i\rangle}
\newcommand{\hs}{\hat{\sigma}}

\begin{document}
\title[CLT for the RFIM free energy]{Central limit theorem for the free energy of  the  random field Ising model}
\author{Sourav Chatterjee}
\address{\newline Department of Statistics \newline Stanford University\newline Sequoia Hall, 390 Serra Mall \newline Stanford, CA 94305\newline \newline \textup{\tt souravc@stanford.edu}}
\thanks{Research partially supported by NSF grant DMS-1608249}
\keywords{Random field Ising model, central limit theorem, free energy}
\subjclass[2010]{82B44, 60K35}

\begin{abstract}
A central limit theorem is proved for the free energy of the random field Ising model with all plus or all minus boundary condition, at any temperature (including zero temperature) and any dimension. This solves a problem posed by Wehr and Aizenman in 1990. The proof uses a variant of Stein's method. 
\end{abstract}

\maketitle

\section{Introduction}
Take any $d\ge 1$. Let $\Lambda$ be any finite subset of $\zz^d$. Let $\partial \Lambda$ be the set of all $x\in \zz^d\setminus \Lambda$ that are adjacent to some point of $\Lambda$. Let $\Sigma=\{-1,1\}^\Lambda$, $\Gamma = \{-1,1\}^{\partial\Lambda}$ and $\Phi = \rr^\Lambda$. Given $\sigma \in \Sigma$, $\gamma \in \Gamma$ and $\phi\in \Phi$, define 
\[
H_{\gamma, \phi}(\sigma) := -\frac{1}{2}\sum_{\substack{x,y\in \Lambda,\\x\sim y}} \sigma_x\sigma_y - \sum_{\substack{x\in \Lambda, \, y\in \partial \Lambda,\\ x\sim y}}\sigma_x\gamma_y- \sum_{x\in \Lambda} \phi_x \sigma_x,
\]
where $x\sim y$ means that $x$ and $y$ are neighbors. Take any $\beta\in [0,\infty]$. The Ising model on $\Lambda$ with boundary condition $\gamma$, inverse temperature $\beta$, and external field $\phi$, is the probability measure on $\Sigma$ with probability mass function proportional to $e^{-\beta H_{\gamma, \phi}(\sigma)}$. When $\beta=\infty$, this is simply the uniform probability measure on the configurations that minimize the energy. 

Now suppose that $(\phi_x)_{x\in \Lambda}$ are i.i.d.~random variables instead of fixed constants. Then the probability measure defined above becomes a random probability measure. This is known as the random field Ising model (RFIM). We will refer to the law of $\phi_x$ as the random field distribution. 

The random field Ising model was introduced by \citet{imryma75} in 1975 as a simple example of a disordered system. Imry and Ma predicted that the model does not have an ordered phase in dimensions one and two, but does exhibit a phase transition in dimensions three and higher. The existence of the phase transition in dimension three was partially proved by \citet{imbrie84, imbrie85}, who showed that there are two macroscopic ground states in the 3D RFIM. The phase transition at nonzero temperature was finally established by  \citet{bk87, bk88} in 1987, settling the Imry--Ma conjecture in $d\ge 3$. A few years later, \citet{aw89, aw90} proved  the non-existence of an ordered phase in $d\le 2$. The proof of the Imry--Ma conjecture is regarded as a notable success story of mathematical physics, because there was considerable debate within the theoretical physics community about the validity of the conjecture. See \citet[Chapter 7]{bovier06} for more details.

Another important paper on the random field Ising model is the 1990 paper of~\citet{wehraizenman90} on the fluctuations of the free energy of the RFIM and related models. The free energy of the RFIM on $\Lambda$ at inverse temperature $\beta$ and boundary condition $\gamma$ is defined as 
\[
F(\gamma, \phi,\beta) := -\frac{1}{\beta} \log \sum_{\sigma \in \Sigma} e^{-\beta H_{\gamma,\phi}(\sigma)}. 
\]
When $\beta =\infty$, the free energy is simply the ground state energy:
\[
F(\gamma, \phi, \infty) = \min_{\sigma\in \Sigma} H_{\gamma,\phi}(\sigma).
\]
\citet{wehraizenman90} proved that under mild conditions on the random field distribution, the variance of $F$ is upper and lower bounded by constant multiples of the size of $\Lambda$. In the same paper, Wehr and Aizenman posed the problem of proving a central limit theorem for $F$  as the size of $\Lambda$ tends to infinity. The main result of this paper is a solution of this question for the RFIM with plus or minus boundary condition. The plus boundary condition is the boundary condition $\gamma$ where $\gamma(x)=+1$ for all $x\in \partial \Lambda$. Similarly, the minus boundary condition has $\gamma(x)=-1$ for all $x\in \partial\Lambda$. When $d\le 2$, the result holds for any boundary condition.

Results about fluctuations of the free energy have a number of applications. As stated in \cite{wehraizenman90}, bounds on fluctuations of the free energy were instrumental in the proof of rounding effects of the quenched randomness on first-order phase transitions in low-dimensional systems. Another application in a different model, also discussed in \cite{wehraizenman90}, is an inequality for characteristic exponents of the model of directed polymers in a random environment. Central limit theorems give the most precise information about fluctuations, and they are also mathematically interesting in their own right. Central limit theorems for the free energy have been proved for disordered systems with mean-field interactions such as the Sherrington--Kirkpatrick model of spin glasses~\cite{alr87, chl18, chenetal17}. But as far as I am aware, no such results were available for disordered systems on lattices prior to this paper.

The main result has two parts, corresponding to the cases $\beta<\infty$ and $\beta=\infty$. The $\beta<\infty$ case is the following. 
\begin{thm}\label{finite}
Take any $d\ge 1$. Let $\{\Lambda_n\}_{n\ge 1}$ be a sequence of finite nonempty subsets of $\zz^d$. For each $n$, consider the RFIM on $\Lambda_n$ with plus boundary condition, at inverse temperature $\beta \in (0,\infty)$. Suppose that the random field distribution has finite moment generating function. Let $F_n$ be the free energy of the model. Suppose that $|\partial\Lambda_n|= o(|\Lambda_n|)$ as $n\to \infty$. Then there is a finite positive constant $\sigma^2$, depending only on $\beta$, $d$ and the random field distribution (and not on the sequence $\{\Lambda_n\}_{n\ge 1}$), such that 
\[
\lim_{n\to\infty} \frac{\var(F_n)}{|\Lambda_n|}= \sigma^2,
\]
and
\[
\frac{F_n- \ee(F_n)}{\sqrt{|\Lambda_n|}} \stackrel{\dm}{\to} \mathcal{N}(0,\sigma^2),
\]
where $\stackrel{\dm}{\to}$ denotes convergence in law, and $\mathcal{N}(0,\sigma^2)$ is the normal distribution with mean zero and variance $\sigma^2$. The same result holds for minus boundary condition, possibly with a different value of $\sigma^2$. If $d\le 2$, then the above conclusion holds under any arbitrary sequence of boundary conditions. 
\end{thm}
Note that in Theorem \ref{finite}, the only condition that we imposed on the random field distribution is that it has finite moment generating function. For the $\beta=\infty$ case, our proof technique requires that (a) the random field distribution is continuous, and (b) it is a push-forward of the standard normal distribution under a Lipschitz map (with arbitrary Lipschitz constant). For example, the normal distribution with any mean and any variance belongs to this class. The uniform distribution on any interval is another example. 
\begin{thm}\label{infinite}
Take any $d\ge 1$ and let $\{\Lambda_n\}_{n\ge 1}$ be a sequence of finite subsets of $\zz^d$. Suppose, as in Theorem \ref{finite}, that $|\partial \Lambda_n|= o(|\Lambda_n|)$ as $n\to \infty$. Let $G_n$ be the ground state energy of the RFIM on $\Lambda_n$ with plus boundary condition. Suppose that the random field distribution satisfies the conditions stated above. Then there is a finite positive constant $\sigma^2$, depending only on $d$ and the random field distribution (and not on the sequence $\{\Lambda_n\}_{n\ge 1}$), such that 
\[
\lim_{n\to\infty}\frac{\var(G_n)}{|\Lambda_n|}= \sigma^2,
\]
and  
\[
\frac{G_n-\ee(G_n)}{\sqrt{|\Lambda_n|}} \stackrel{\dm}{\to} \mathcal{N}(0,\sigma^2).
\]
The same result holds for minus boundary condition, possibly with a different value of $\sigma^2$. If $d\le 2$, then the above conclusion holds under any arbitrary sequence of boundary conditions.
\end{thm}
The main tool for proving Theorem \ref{finite} is a method of normal approximation introduced in \cite{cha08}, where it was developed as an extension of Stein's method~\cite{stein72, stein86}. A `continuous' version of this method, developed in \cite{cha09a}, is our tool for proving Theorem \ref{infinite}. The extension to arbitrary boundary conditions in $d\le 2$ is possible because of the uniqueness of the infinite volume Gibbs state in $d\le 2$, famously proved by \citet{aw90}. A quantitative version of the Aizenman--Wehr result, such as the ones recently proved in \cite{chatterjee17} and \cite{ap18}, can be used to obtain rates of convergence in Theorems \ref{finite} and \ref{infinite} when $d\le 2$. In particular, if the rate from \cite{ap18} is used, then it should be possible to prove a rate of convergence of order $n^{-\alpha}$ for some small positive constant $\alpha$ using the methods of this paper.  

There are several questions that remain open about central limit theorems for the RFIM. The foremost is proving (or disproving) central limit theorems under arbitrary boundary conditions in $d\ge 3$. The main technical difficulty is that for arbitrary boundary conditions, it is not clear how to establish a result like inequality \eqref{lambdadiff} of Section \ref{finiteproof}, which is crucial for the proof. 

Another problem is to express the limiting variance $\sigma^2$ in some kind of a closed form, instead of just saying that it exists. The problem of getting any rate of convergence in $d\ge 3$ is also interesting and beyond the reach of existing ideas.

Lastly, one may wonder if the methods of this paper can be applied to prove CLTs in other disordered systems on lattices, such as the Ising spin glass. In principle the method should work as long as a decay of correlation result like inequality \eqref{lambdadiff} of Section \ref{finiteproof} can be established. At present, however, it is not known how to establish decay of correlations in the Ising spin glass except at high temperature.

\section{Technique}
First, let us briefly review the main result of \cite{cha08}. Recall that the Wasserstein distance $d_{\textup{W}}(\mu,\nu)$ between the two probability measures $\mu$ and $\nu$ on $\rr$ is defined to be the supremum of $|\int fd\mu-\int fd\nu|$ over all Lipschitz $f:\rr\to \rr$ with Lipschitz constant $1$.

Let $\xx$ be a measurable space and suppose that $X= (X_1,\ldots,X_n)$ is a vector of independent $\xx$-valued random variables. Let $\xp = (\xp_1,\ldots,\xp_n)$ be an independent copy of $X$. Let $[n] = \{1,\ldots, n\}$, and for each $A \subseteq [n]$, define the random vector $X^A$ as
\[
X^A_i =
\begin{cases}
\xp_i &\text{ if } i\in A,\\
X_i &\text{ if } i\not \in A.
\end{cases}
\]
For each $i$, let
\[
\Delta_i f(X) := f(X)-f(X^{\{i\}}),
\]
and for each $A\subseteq[n]$ and $i\not\in A$, let
\[
\Delta_i f(X^A) := f(X^A)-f(X^{A\cup\{i\}}). 
\]
Let
\begin{equation*}
T:= \frac{1}{2}\sum_{i=1}^n \sum_{A\subseteq [n]\setminus\{i\}} \frac{\Delta_i f(X)\Delta_i f(X^A)}{n{n-1 \choose |A|}}.
\end{equation*}
The following theorem is the main result of \cite{cha08}.
\begin{thm}[\cite{cha08}]\label{normthm} 
Let all terms be defined as above, and let $W = f(X)$. Suppose that $W$ has finite second moment, and let $\sigma^2:= \var(W)$. Let $\mu$ be the law of $(W-\ee(W))/\sigma$ and $\nu$ be the standard normal distribution on the real line. Then $\ee(T) = \sigma^2$ and
\[
d_{\textup{W}}(\mu,\nu) \le \frac{\sqrt{\var(\ee(T|W))}}{\sigma^2} + \frac{1}{2\sigma^3}\sum_{i=1}^n \ee|\Delta_if(X)|^3.
\]
\end{thm}
Recall that the Kolmogorov distance between two probability measures $\mu$ and $\nu$ on the real line is defined as
\[
d_{\textup{K}}(\mu, \nu) := \sup_{t\in \rr} |\mu((-\infty, t]) - \nu((-\infty, t])|.
\]
The Kolmogorov distance is more commonly used in probability and statistics than the Wasserstein distance. The bound on $d_{\textup{W}}(\mu,\nu)$ in Theorem \ref{normthm} can be used to get a bound on $d_{\textup{K}}(\mu,\nu)$ using the following simple observation made in \citet{chasound}: Let $\nu$ denote the standard normal distribution and let $\mu$ be any  probability measure on $\rr$. Then 
\begin{equation} 
\label{KKW}  
d_{\textup{K}}(\mu,\nu) \le 2\sqrt{d_{\textup{W}}(\mu,\nu)}.
\end{equation} 
The combination of Theorem \ref{normthm} and inequality \eqref{KKW} usually yields a suboptimal bound for the Kolmogorov distance. There is a recent improvement of Theorem \ref{normthm} by \citet{lp17} that gives optimal bounds for the Kolmogorov distance in many problems. 

Theorem \ref{normthm} by itself is a bit difficult to directly apply to the problem at hand. We will now synthesize a corollary of Theorem \ref{normthm} that will be more easily applicable for the random field Ising model. The main idea here is to approximate the discrete derivative $\Delta_i f(X)$ by a function that depends `on only a few coordinates'. We will continue to work in the setting introduced above. 

For each $1\le i\le n$, let $g_i:\mx^{n}\times \mx \to \rr$ be a measurable map. For each $i$ and each $p\ge 1$, let
\[
m_{p,i} := \|\Delta_i f(X)\|_{L^p} = (\ee|\Delta_i f(X)|^p)^{1/p}
\]
and let
\[
\ep_{p,i} := \|\Delta_i f(X)- g_i(X, X_i')\|_{L^p}. 
\]
First, we have the following generalization of Theorem \ref{normthm}.
\begin{thm}\label{normgen}
Let all notation be as above, and let $\mu$, $\nu$ and $\sigma$ be as in Theorem \ref{normthm}. Let
\[
S := \frac{1}{2}\sum_{i=1}^n \sum_{A\subseteq [n]\setminus\{i\}} \frac{g_i(X, X_i') g_i(X^A, X_i')}{n{n-1 \choose |A|}}.
\]
Then
\[
|\sigma^2-\ee(S)|\le \frac{1}{2}\sum_{i=1}^n (2\ep_{2,i}m_{2,i}+ \ep_{2,i}^2), 
\]
and 
\begin{align*}
d_{\textup{W}}(\mu, \nu) &\le \frac{1}{2\sigma^2}\sum_{i=1}^n (2\ep_{4,i}m_{4,i}+ \ep_{4,i}^2) + \frac{\sqrt{\var{S}}}{\sigma^2}  + \frac{1}{2\sigma^3}\sum_{i=1}^n m_{3,i}^3.
\end{align*}
\end{thm}
\begin{proof}
For simplicity of notation, let $Y_i := g_i(X,X_i)$ and $Y_i^A := g_i(X^A, X_i')$. Note that for any $A$ and any $i\not \in A$, 
\begin{align}
&\|\Delta_i f(X)\Delta_i f(X^A) - Y_iY_i^A\|_{L^2} \nonumber\\
&\le \|(\Delta_i f(X)- Y_i)\Delta_i f(X^A)\|_{L^2}+ \|Y_i(\Delta_i f(X^A)- Y_i^A)\|_{L^2}\nonumber \\
&\le \|\Delta_i f(X)-Y_i\|_{L^4} \|\Delta_i f(X^A)\|_{L^4} + \|Y_i\|_{L^4}\|\Delta_i f(X^A)- Y_i^A\|_{L^4}\nonumber \\
&= \ep_{4,i}(m_{4,i} + \|Y_i\|_{L^4})\nonumber \\
&\le \ep_{4,i}(m_{4,i} + \|Y_i-\Delta_if(X)\|_{L^4}+ \|\Delta_i f(X)\|_{L^4})\nonumber \\
&= 2\ep_{4,i}m_{4,i}+ \ep_{4,i}^2.\label{moment}
\end{align}
Now, for each $i$,
\begin{equation}\label{identity}
\sum_{A\subseteq [n]\setminus\{i\}} \frac{1}{n{n-1 \choose |A|}} = \sum_{k=0}^{n-1}\frac{|\{A: A\subseteq [n]\setminus\{i\}, |A|=k\}|}{n{n-1 \choose k}} = 1.
\end{equation}
Therefore,
\eq{
\|T-S\|_{L^2} &\le \frac{1}{2} \sum_{i=1}^n \sum_{A\subseteq [n]\setminus\{i\}} \frac{\|\Delta_i f(X)\Delta_i f(X^A) - Y_i Y_i^A\|_{L^2}}{n{n-1 \choose |A|}}\\
&\le \frac{1}{2}\sum_{i=1}^n \sum_{A\subseteq [n]\setminus\{i\}} \frac{2\ep_{4,i}m_{4,i}+ \ep_{4,i}^2}{n{n-1 \choose |A|}}\\
&= \frac{1}{2}\sum_{i=1}^n (2\ep_{4,i}m_{4,i}+ \ep_{4,i}^2).
}
Consequently,
\eq{
\sqrt{\var(T)} &\le \|T-\ee(S)\|_{L^2}\\
&\le \|T-S\|_{L^2} + \sqrt{\var(S)}
\\
&\le \frac{1}{2}\sum_{i=1}^n (2\ep_{4,i}m_{4,i}+ \ep_{4,i}^2) + \sqrt{\var(S)}. 
}
This bound, together with Theorem \ref{normthm} and the observation that 
\[
\var(\ee(T|W))\le \var(T),
\]
gives the second inequality in the statement of Theorem \ref{normgen}. For the first inequality, recall from Theorem \ref{normthm} that $\ee(T)=\sigma^2$. Then retrace the steps in the derivation of \eqref{moment} starting with the $L^1$ norm instead of the $L^2$ norm, and  finally use the identity \eqref{identity}, to get
\eq{
|\ee(T)-\ee(S)| &\le \frac{1}{2}\sum_{i=1}^n \sum_{A\subseteq [n]\setminus\{i\}} \frac{\ee|\Delta_i f(X)\Delta_i f(X^A)- Y_i Y_i^A|}{n{n-1 \choose |A|}}\\
&\le \frac{1}{2}\sum_{i=1}^n \sum_{A\subseteq [n]\setminus\{i\}} \frac{2\ep_{2,i} m_{2,i} + \ep_{2,i}^2}{n{n-1 \choose |A|}}\\
&= \frac{1}{2}\sum_{i=1}^n (2\ep_{2,i} m_{2,i} + \ep_{2,i}^2). 
}
This completes the proof of the theorem. 
\end{proof}
Theorem \ref{normgen} can be useful only when it is easier to understand $\var(S)$ than $\var(T)$. The following result gives such a criterion.
\begin{prop}\label{giprop}
Let $g_i$ and $S$ be as in Theorem \ref{normgen}. Suppose that for each $i$, there is a set $N_i\subseteq [n]$ such that $g_i(x, x_i')$ is a function of only $(x_j)_{j\in N_i}$ and $x_i'$. Then
\eq{
\var(S) &\le \frac{1}{4} \sum_{i,j: N_i\cap N_j \ne \emptyset} (m_{4,i} + \ep_{4,i})^2(m_{4,j}+ \ep_{4,j})^2.
}
\end{prop}
\begin{proof}
Let $Y_i$ and $Y_i^A$ be as in the proof of Theorem \ref{normgen}. Note that
\eq{
\var(S) &= \frac{1}{4}\sum_{i,j=1}^n \sum_{\substack{A\subseteq [n]\setminus \{i\},\\ B \subseteq [n]\setminus \{j\}}}\frac{\cov(Y_iY_i^A, \, Y_j Y_j^B)}{n^2{n-1\choose |A|} {n-1\choose |B|}}. 
}
By independence of coordinates, whenever $N_i\cap N_j = \emptyset$, 
\[
\cov(Y_i Y_i^A,\, Y_jY_j^B) = 0.
\]
Moreover, for any $i$, $j$, $A$ and $B$,
\eq{
|\cov(Y_i Y_i^A,\, Y_jY_j^B)|&\le \|Y_i Y_i^A\|_{L^2} \|Y_j Y_j^B\|_{L^2}\\
&\le \|Y_i\|_{L^4}^2\|Y_j\|_{L^4}^2\\
&\le (m_{4,i} + \ep_{4,i})^2(m_{4,j}+ \ep_{4,j})^2. 
}
By \eqref{identity}, this completes the proof. 
\end{proof}
Combining Theorem \ref{normgen} and Proposition \ref{giprop}, we get the following result. This is our main tool for proving Theorem \ref{finite}. 
\begin{thm}\label{normcomb}
Let all notation be as in Theorem \ref{normgen}. Let $N_i$ be as in Proposition \ref{giprop}. Then 
\[
|\sigma^2-\ee(S)|\le \frac{1}{2}\sum_{i=1}^n (2\ep_{2,i}m_{2,i}+ \ep_{2,i}^2), 
\]
and 
\begin{align*}
d_{\textup{W}}(\mu, \nu) &\le \frac{1}{2\sigma^2}\sum_{i=1}^n (2\ep_{4,i}m_{4,i}+ \ep_{4,i}^2) \\
&\qquad + \frac{1}{2\sigma^2}\biggl( \sum_{i,j: N_i\cap N_j \ne \emptyset} (m_{4,i} + \ep_{4,i})^2(m_{4,j}+ \ep_{4,j})^2\biggr)^{1/2} \\
&\qquad  + \frac{1}{2\sigma^3}\sum_{i=1}^n m_{3,i}^3.
\end{align*}
\end{thm}
Theorem \ref{normcomb} will be used in Section \ref{finiteproof} to prove Theorem \ref{finite}. However, I have not been able to use Theorem \ref{normcomb} to prove Theorem \ref{infinite} (the CLT for the ground state energy). Instead, a `continuous version' of Theorem \ref{normcomb} will be used to prove Theorem \ref{infinite}. This is presented as Theorem~\ref{normcont} below.

Let $f:\rr^n \to \rr$ be a differentiable function.  Let $\partial_i f$ denote the partial derivative of $f$ in the $i^{\textup{th}}$ coordinate, and let $\nabla f = (\partial_1 f,\ldots, \partial_n f)$ be the gradient of $f$. Let $Z = (Z_1,\ldots,Z_n)$ be a vector of i.i.d.\ standard normal random variables. The main ingredient in the proof of Theorem~\ref{normcont} is the following lemma, which is a slightly modified version of Lemma 5.3 from \citet{cha09a}. Recall that the total variation distance between two probability measures $\mu$ and $\nu$ on the real line is defined as
\[
\tv(\mu,\nu) := \sup_A|\mu(A)-\nu(A)|,
\]
where the supremum is taken over all Borel subsets of $\rr$. 
\begin{lmm}[\cite{cha09a}]\label{gaussian}
Let $f$ and $Z$ be as in the above paragraph and let $W:=f(Z)$. Assume that $\|f(Z)\|_{L^4}<\infty$ and $\|\partial_i f(Z)\|_{L^4}<\infty$  for all $i$. Let $\sigma^2:=\var(W)$. Let $Z'$ be an independent copy of $Z$, and let
\[
T := \int_0^1 \frac{1}{2\st} \nabla f(Z)\cdot \nabla f (\st Z + \sst Z') dt.
\]
Let $\mu$ be the law of $(W-\ee(W))/\sigma$ and $\nu$ be the standard normal distribution. Then $\ee(T)=\sigma^2$ and  
\[
\tv(\mu,\nu) \le  \frac{2\sqrt{\var(T)}}{\sigma^2}.
\]
\end{lmm}
The above lemma is the starting point for the method of `second order Poincar\'e inequalities' developed in \cite{cha09a}. For proving the CLT for the ground state energy of the RFIM, however, I could not construct a proof using  second order Poincar\'e inequalities. Instead, the above lemma needs to be used in a different way, more along the lines of Theorem \ref{normcomb}. 

For each $1\le i\le n$, let $g_i:\rr^n \to \rr$ be a measurable function and let $N_i$ be a set of coordinates such that the value of $g_i(x_1,\ldots, x_n)$ is determined by $(x_j)_{j\in N_i}$. Suppose that $\|g_i(Z)\|_{L^4}<\infty$ for all $i$. For each $1\le i\le n$ and $p\ge 1$, let
\begin{equation}\label{mpdef}
m_{p,i} := \|\partial_i f(Z)\|_{L^p}
\end{equation}
and
\begin{equation}\label{epdef}
\ep_{p,i} := \|\partial_i f(Z)-g_i(Z)\|_{L^p}.
\end{equation}
Let $g:\rr^n\to \rr^n$ be the function whose $i^{\mathrm{th}}$ coordinate map is $g_i$.  For $0\le t\le 1$, let
\[
Z^t := \st Z +\sst Z',
\]
and let 
\begin{equation}\label{sdef}
S:= \int_0^1 \frac{1}{2\st}g(Z)\cdot g(Z^t)dt. 
\end{equation}
The following theorem gives a continuous analog of Theorem \ref{normcomb}, in the setting of Lemma \ref{gaussian}.
\begin{thm}\label{normcont}
Let $\ep_{p,i}$ and $m_{p,i}$ be defined as above and let all other variables be defined as in Lemma \ref{gaussian}. Then 
\[
|\sigma^2-\ee(S)|\le \sum_{i=1}^n (2\ep_{2,i}m_{2,i}+\ep_{2,i}^2),
\]
and
\begin{align*}
\tv(\mu,\nu) &\le \frac{2}{\sigma^2}\sum_{i=1}^n (2\ep_{4,i}m_{4,i}+\ep_{4,i}^2) \\
&\qquad + \frac{2}{\sigma^2}\biggl(\sum_{i,j: N_i \cap N_j\ne \emptyset} (m_{4.i}+\ep_{4,i})^2(m_{4,j}+\ep_{4,j})^2\biggr)^{1/2}.
\end{align*}
\end{thm}
\begin{proof}
Note that
\begin{align*}
\|T-S\|_{L^2} &\le \int_0^1 \frac{1}{2\st}\|\nabla f (Z)\cdot \nabla f(Z^t) - g(Z)\cdot g(Z^t)\|_{L^2}dt.
\end{align*}
But for any $t$, 
\begin{align*}
&\|\nabla f (Z)\cdot \nabla f(Z^t) - g(Z)\cdot g(Z^t)\|_{L^2}\\
&\le \|(\nabla f(Z) - g(Z))\cdot\nabla f(Z^t)\|_{L^2} + \|g(Z)\cdot(\nabla f(Z^t) - g(Z^t))\|_{L^2}\\
&\le \sum_{i=1}^n (\|(\partial_i f(Z) - g_i(Z))\partial_i f(Z^t)\|_{L^2}+ \|g_i(Z)(\partial_i f(Z^t) - g_i(Z^t))\|_{L^2})\\
&\le \sum_{i=1}^n (\|\partial_i f(Z) - g_i(Z)\|_{L^4}\|\partial_i f(Z^t)\|_{L^4}+ \|g_i(Z)\|_{L^4}\|\partial_i f(Z^t) - g_i(Z^t)\|_{L^4})\\
&= \sum_{i=1}^n (\ep_{4,i}m_{4,i}+\|g_i(Z)\|_{L^4}\ep_{4,i})\\
&\le \sum_{i=1}^n (\ep_{4,i}m_{4,i}+(m_{4,i}+\ep_{4,i})\ep_{4,i}).
\end{align*}
Thus,
\begin{equation}\label{teq}
\|T-S\|_{L^2}\le \sum_{i=1}^n (2\ep_{4,i}m_{4,i}+\ep_{4,i}^2). 
\end{equation}
On the other hand, 
\begin{align}
\sqrt{\var(T)} &\le \|T-\ee(S)\|_{L^2}\nonumber\\
&\le \|T-S\|_{L^2} + \sqrt{\var(S)}. \label{vv}
\end{align}
By Jensen's inequality,
\begin{align*}
\var(S) &= \ee\biggl(\int_0^1\frac{1}{2\st} (g(Z)\cdot g(Z^t)-\ee(g(Z)\cdot g(Z^t))) dt\biggr)^2\\
&\le \int_0^1\frac{1}{2\st} \var(g(Z)\cdot g(Z^t))dt\\
&= \int_0^1 \frac{1}{2\st} \sum_{i,j=1}^n \cov(g_i(Z)g_i(Z^t), \,g_j(Z)g_j(Z^t))dt. 
\end{align*}
Now note that if $N_i\cap N_j=\emptyset$, then 
\[
\cov(g_i(Z)g_i(Z^t), \,g_j(Z)g_j(Z^t))=0,
\]
and for any $i$ and $j$, 
\eq{
\cov(g_i(Z)g_i(Z^t), \,g_j(Z)g_j(Z^t)) &\le \|g_i(Z)g_i(Z^t)\|_{L^2}\|g_j(Z)g_j(Z^t)\|_{L^2}\\
&\le \|g_i(Z)\|_{L^4}^2\|g_j(Z)\|_{L^4}^2\\
&\le (m_{4,i}+\ep_{4,i})^2(m_{4,j}+\ep_{4,j})^2.  
}
This shows that 
\eq{
\var(S) \le \sum_{i,j: N_i \cap N_j\ne \emptyset} (m_{4.i}+\ep_{4,i})^2(m_{4,j}+\ep_{4,j})^2.
}
Combining this with \eqref{teq}, \eqref{vv} and Lemma \ref{gaussian}, we get the desired bound on $\tv(\mu,\nu)$. For the bound on $|\sigma^2-\ee(S)|$, we proceed as in the proof of~\eqref{teq} to obtain a bound on $\|T-S\|_{L^1}$,  and then use Lemma \ref{gaussian} for the identity $\sigma^2=\ee(T)$.
\end{proof}
Theorem \ref{normcont} will be used to prove Theorem \ref{infinite} in Section \ref{infiniteproof}. In that proof, $f$ will be the ground state energy of the RFIM on a finite set, considered as a function of the random field. However,  it is not a differentiable function of the random field. To take care of this issue, we need to extend Theorem~\ref{normcont} to the slightly larger class of functions. 
\begin{prop}\label{normprop}
For each $k$, let  $f_k:\rr^n \to\rr$ be a differentiable function. Suppose that $f(x)=\lim_{k\to\infty} f_k(x)$ exists almost everywhere. Further, assume that for each $i$, $\lim_{k\to\infty} \partial_i f_k(x)$ exists almost everywhere, and call the limit $\partial_i f(x)$. Lastly, suppose that for some $\ep>0$,
\begin{equation}\label{unifint}
\sup_k\|f_k(Z)\|_{L^{4+\ep}}<\infty \ \text{ and } \ \sup_{i,k}\|\partial_i f_k(Z)\|_{L^{4+\ep}}<\infty,
\end{equation}
where $Z = (Z_1,\ldots,Z_n)$ is a vector of i.i.d.~standard normal random variables.
Take any $g_1,\ldots, g_n$ as in the paragraph preceding the statement of Theorem~\ref{normcont}, and define $m_{p,i}$ and $\ep_{p,i}$ as in \eqref{mpdef} and \eqref{epdef}, assuming that $\|g_i(Z)\|_{L^{4+\ep}}<\infty$ for each $i$. Then the conclusions of Theorem~\ref{normcont} hold for the function $f$, treating $\partial_i f$ as its derivative in the $i^{th}$ coordinate. 
\end{prop}
\begin{proof}
Let $W_k := f_k(Z)$, $\sigma_k^2:= \var(W_k)$, and $\mu_k$ be the law of $(W_k-\ee(W_k))/\sigma_k$. Let $S$ be defined as in \eqref{sdef}. Let $\nu$ be the standard normal distribution. Then Theorem \ref{normcont} gives upper bounds on $|\sigma_k^2-\ee(S)|$ and $\tv(\mu_k, \nu)$ in terms of the $L^2$ and $L^4$ norms of $\partial_i f_k(Z)$ and $\partial_i f_k(Z) - g_i(Z)$. As $k\to \infty$, the a.e.~convergence of $\partial_i f_k$ to $\partial_i f$ and the condition \eqref{unifint} ensure that these norms converge to the corresponding norms of $\partial_i f(Z)$ and $\partial_i f(Z) - g_i(Z)$. This immediately implies the validity of the first inequality of Theorem \ref{normcont} for the function $f$.  

Next, note that the a.e.~convergence of $f_k$ to $f$ and the condition~\eqref{unifint} ensure that $(W_k - \ee(W_k))/\sigma_k$ converges almost surely to $(W-\ee(W))/\sigma$ as $k\to \infty$. This implies that $\mu_k$ converges to $\mu$ weakly. By the well-known coupling characterization of total variation distance, for each $k$ there exists a probability measure $\gamma_k$ on $\rr^2$ whose one-dimensional  marginals are $\mu_k$ and $\nu$, and 
\[
\gamma_k(V)=\tv(\mu_k, \nu),
\]
where
\[
V := \{(x,y)\in \rr^2: x\ne y\}. 
\]
Since $\mu_k$ converges weakly to $\mu$, it follows that the sequence $\{\gamma_k\}_{k\ge 1}$ is a tight family of probability measures on $\rr^2$. Let $\{\gamma_{k_j}\}_{j\ge 1}$ be a subsequence converging to a limit $\gamma$. Then $\gamma$ has marginals $\mu$ and $\nu$. Moreover, since $V$ is an open set,
\eq{
\tv(\mu,\nu)\le \gamma(V) &\le \liminf_{j\to \infty} \gamma_{k_j}(V) \\
&= \liminf_{j\to\infty} \tv(\mu_{k_j}, \nu).  
}
This completes the proof of the proposition.
\end{proof}

\section{Proof of Theorem \ref{finite}}\label{finiteproof}
In this proof, $C$ will denote any positive constant that depends only on $\beta$, $d$ and the random field distribution. The value of $C$ may change from line to line or even within a line. 

We will prove the result under the plus boundary condition only, since the argument for the minus boundary condition is the same. Fix an inverse temperature $\beta$. Let $\si_{\Lambda, \gamma}$ denote the expected value of the spin at site $i$ under the RFIM on $\Lambda$ with boundary condition $\gamma$, at inverse temperature $\beta$. By the FKG property of the random field Ising model, it is a standard fact that for any $\Lambda$ and any $i\in \Lambda$,  $\si_{\Lambda, \gamma}$ is a monotone increasing function of the boundary condition $\gamma$. From this and the Markovian nature of the model, it follows that $\si_{\Lambda, +} \ge \si_{\Lambda',+}$ whenever $i\in \Lambda \subseteq \Lambda'$. 

Take any $i\in \zz^d$. For each $k$, let $\Lambda_{i,k}$ be the cube of side-length $2k+1$ centered at $i$. Then the above inequality shows that the limit
\[
\si_+ := \lim_{k\to \infty} \si_{\Lambda_{i,k}, +}
\]
exists. Therefore, if we let
\begin{equation}\label{deltadef}
\delta_k := \ee|\si_{\Lambda_{i,k},+} - \si_+|,
\end{equation}
then by translation-invariance, $\delta_k$ depends only on $k$ and not on $i$, and 
\[
\lim_{k\to \infty} \delta_k = 0. 
\]
(Note that the absolute value in \eqref{deltadef} is unnecessary, since the random variable inside is nonnegative. But we keep it anyway, to emphasize the point that $\si_{\Lambda_{i,k},+} \approx \si_+$ with high probability when $k$ is large.) Moreover, given any $k$ and $\Lambda$ such that $\Lambda_{i,k}\subseteq\Lambda$, 
\[
\si_+ \le \si_{\Lambda, +}\le \si_{\Lambda_{i,k},+}. 
\]
Consequently,
\begin{equation}\label{lambdadiff}
\ee|\si_{\Lambda, +} - \si_{\Lambda_{i,k},+}|\le \delta_k.
\end{equation}
Now take any nonempty set $\Lambda\subseteq \zz^d$. Fix $\beta$ and let $F$ be the free energy of the RFIM on $\Lambda$ with plus boundary condition, at inverse temperature $\beta$. Consider $F$ as a function of the random field $(\phi_i)_{i\in \Lambda}$, and let $\Delta_i F$ be the change in the value of $F$ when $\phi_i$ is replaced by an independent copy $\phi_i'$, as in Theorem \ref{normthm}. Let
\[
\alpha_i := \beta (\phi_i'-\phi_i).
\]
Then note that
\eq{
\Delta_i F &=-\frac{1}{\beta} \log \smallavg{e^{\alpha_i\sigma_i}}_{\Lambda, +}\\
&= -\frac{1}{\beta}\log \smallavg{\cosh\alpha_i + \sigma_i\sinh\alpha_i}_{\Lambda, +}\\
&= -\frac{1}{\beta}\log (\cosh\alpha_i + \si_{\Lambda, +}\sinh\alpha_i). 
}
In particular, 
\begin{equation}\label{deltaf}
\|\Delta_i F\|_{L^4} \le \frac{\|\alpha_i\|_{L^4}}{\beta} \le C.
\end{equation}
Now fix some $k\ge 1$. For each $i\in \Lambda$, let 
\[
N_i := \Lambda_{i,k}\cap \Lambda.
\]
Let 
\[
g_i := -\frac{1}{\beta}\log (\cosh\alpha_i + \si_{N_i, +}\sinh\alpha_i).
\]
Clearly,
\begin{equation}\label{gibd}
\|g_i\|_{L^4}\le C. 
\end{equation}
For any $x\in [-1,1]$, the quantity  $\cosh\alpha_i + x\sinh\alpha_i$ lies between the numbers $e^{-\alpha_i}$ and $e^{\alpha_i}$. The derivative of the  logarithm function in this interval is bounded above by $e^{|\alpha_i|}$. Therefore, for any $x,y\in[-1,1]$,
\[
|\log(\cosh\alpha_i + x\sinh\alpha_i) - \log(\cosh\alpha_i + y\sinh\alpha_i)|\le e^{|\alpha_i|} |x-y|. 
\]
Thus,
\eq{
|\Delta_i F - g_i|\le e^{|\alpha_i|} |\si_{\Lambda,+}-\si_{N_i,+}|,
}
and so
\eq{
\|\Delta_i F-g_i\|_{L^4} &\le \|e^{|\alpha_i|}\|_{L^8} \|\si_{\Lambda,+}-\si_{N_i,+}\|_{L^8}\\
&\le C (\ee|\si_{\Lambda,+}-\si_{N_i,+}|)^{1/8}.
}
Let $\Lambda'$ be the set of all $i\in \Lambda$ that are at a distance at least $k$ from the boundary of $\Lambda$. Then for each $i\in \Lambda'$, $N_i = \Lambda_{i,k}$, and therefore by \eqref{lambdadiff},
\begin{equation}\label{ep1}
\|\Delta_i F - g_i\|_{L^4}\le C \delta_k^{1/8}. 
\end{equation}
On the other hand, if $i\not \in \Lambda'$, then by \eqref{deltaf} and \eqref{gibd}, 
\begin{equation}\label{ep2}
\|\Delta_i F - g_i\|_{L^4}\le C.
\end{equation}
For each $i\in \Lambda$, note that the number of $j$ such that $N_i\cap N_j \ne \emptyset$ is bounded by $Ck^d$. Also, clearly, 
\begin{equation}\label{llbd}
|\Lambda \setminus \Lambda'|\le Ck^d|\partial \Lambda|.
\end{equation}
Finally, from \cite{wehraizenman90}, we know that 
\begin{equation}\label{sigbd}
\var(F)\ge C |\Lambda|.
\end{equation}
We now have all the estimates required for using Theorem \ref{normcomb}. Let 
\[
m_{p,i} := \|\Delta_i F\|_{L^p}
\]
and 
\[
\ep_{p,i} := \|\Delta_i F - g_i\|_{L^p}. 
\]
By the estimates obtained above, 
\begin{align*}
\sum_{i\in \Lambda} (2\ep_{4,i}m_{4,i}+\ep_{4,i}^2)&\le  C|\Lambda'| \delta_k^{1/8} + C|\Lambda\setminus \Lambda'|\\
&\le C|\Lambda| \delta_k^{1/8} + Ck^{d} |\partial \Lambda|. 
\end{align*}
Next, note that
\begin{align*}
\sum_{i,j: N_i\cap N_j \ne \emptyset} (m_{4,i} + \ep_{4,i})^2(m_{4,j}+ \ep_{4,j})^2&\le Ck^d|\Lambda |.
\end{align*}
Finally, 
\eq{
\sum_{i\in \Lambda} m_{3,i}^3\le C|\Lambda|. 
}
Let $\mu$ denote the law of $(F-\ee(F))/\sqrt{\var(F)}$ and let $\nu$ denote the standard normal distribution. Plugging the above bounds into Theorem \ref{normcomb}, and using the lower bound \eqref{sigbd}, we get
\[
d_{\textup{W}}(\mu, \nu) \le  C \delta_k^{1/8} + Ck^{d}\frac{|\partial \Lambda|}{|\Lambda|} + \frac{Ck^{d/2}}{\sqrt{|\Lambda|}}.
\]
Let $F_n$ and $\Lambda_n$ be as in the statement of the theorem. Let $\mu_n$ be the law of $(F_n-\ee(F_n))/\sqrt{\var(F_n)}$. Since $|\partial \Lambda_n| = o(|\Lambda_n|)$ as $n\to\infty$, the above bound shows that 
\[
\limsup_{n\to\infty} d_{\textup{W}}(\mu_n,\nu)\le C\delta_k^{1/8}. 
\]
However, $k$ is arbitrary, and $\delta_k\to0$ as $k\to \infty$. This shows that $\mu_n$ converges to $\nu$ in the Wasserstein metric.

To complete the proof of Theorem \ref{finite}, it only remains to show that the ratio $\var(F_n)/|\Lambda_n|$ tends to a finite nonzero limit. For this, we will use the first inequality of Theorem \ref{normcomb} and the following simple lemma. 
\begin{lmm}\label{idlmm}
For any integers $m\ge l\ge 0$ and $n\ge 0$,
\[
\sum_{k=0}^n \frac{{n\choose k}}{{n+m\choose k+l}} = \frac{n+m+1}{(m+1){m\choose l}}.
\]
\end{lmm}
\begin{proof}
By the well-known formula for the beta integral,
\begin{align*}
\int_0^1 x^{k+l}(1-x)^{n+m-k-l}dx &= \frac{(k+l)!(n+m-k-l)!}{(n+m+1)!}\\
&= \frac{1}{(n+m+1){n+m\choose k+l}}.
\end{align*}
Thus, again by the beta integral formula,
\begin{align*}
\sum_{k=0}^n \frac{{n\choose k}}{{n+m\choose k+l}}  &= \int_0^1 (n+m+1)\sum_{k=0}^n{n\choose k} x^{k+l}(1-x)^{n+m-k-l}dx \\
&= \int_0^1 (n+m+1)x^l(1-x)^{m-l}dx \\
&= \frac{(n+m+1)l!(m-l)!}{(m+1)!}. 
\end{align*}
This completes the proof of the lemma.
\end{proof}
We will now show that under the conditions of Theorem \ref{finite}, $\var(F_n)/|\Lambda_n|$ tends to a finite nonzero limit. Fix $k\ge 1$ and let $N_i$ and $g_i$ be as before. Consider $g_i$ as a function of $(\phi_j)_{j\in \Lambda}$. For each $A\subseteq \Lambda$ such that $i\not \in A$, let $g_i^A$ be the value of $g_i$ after replacing $\phi_j$ with an independent copy $\phi_j'$ for each $j\in A$. Then the quantity $S$ of Theorem \ref{normcomb} is simply
\[
\frac{1}{2}\sum_{i\in \Lambda} \sum_{A\subseteq \Lambda\setminus \{i\}}\frac{g_i g_i^A}{|\Lambda|{|\Lambda|-1\choose |A|}}.
\]
This can be rewritten as
\[
\frac{1}{2}\sum_{i\in \Lambda} \sum_{A_1\subseteq N_i\setminus \{i\}}\sum_{A_2\subseteq \Lambda \setminus N_i}\frac{g_i g_i^{A_1\cup A_2}}{|\Lambda|{|\Lambda|-1\choose |A_1|+|A_2|}}.
\]
But for any $i$, $A_1$ and $A_2$ as in the above display, the definition of $g_i$ implies that
\[
g_i^{A_1\cup A_2} = g_i^{A_1}.
\]
Thus,
\[
S = \sum_{i\in \Lambda} S_i,
\]
where
\begin{equation}\label{sieq}
S_i := \frac{1}{2}\sum_{A_1\subseteq N_i\setminus \{i\}} g_i g_i^{A_1}\biggl(\sum_{A_2\subseteq \Lambda \setminus N_i}\frac{1}{|\Lambda|{|\Lambda|-1\choose |A_1|+|A_2|}}\biggr).
\end{equation}
Let $p(\Lambda, i,A_1)$ denote the term within the brackets in the above display. Note that by \eqref{identity},
\[
\sum_{A_1\subseteq N_i\setminus \{i\}} p(\Lambda, i,A_1)= \sum_{A\subseteq \Lambda\setminus \{i\}} \frac{1}{|\Lambda|{|\Lambda|-1\choose |A|}} =1. 
\]
Consequently, for any $i\in \Lambda$,
\begin{equation}\label{sibd2}
\ee|S_i| \le \frac{1}{2}\sum_{A_1\subseteq N_i\setminus \{i\}} p(\Lambda, i, A_1)\ee|g_i g_i^{A_1}|\le C. 
\end{equation}
On the other hand, it is not difficult to see from the expression \eqref{sieq} and the definitions of $g_i$, $N_i$ and $\Lambda'$ that $\ee(S_i)$ is the same for all $i\in \Lambda'$. Without loss of generality, suppose that the origin $0$ is in $\Lambda'$. Then by the preceding remark,
\[
\ee(S) = |\Lambda'| \ee(S_0)+\sum_{i\in \Lambda\setminus \Lambda'} \ee(S_i).
\]
By \eqref{sibd2} and \eqref{llbd}, this gives
\begin{equation}\label{ss0}
|\ee(S)-|\Lambda|\ee(S_0)|\le Ck^d|\partial \Lambda|.
\end{equation}
On the other hand, by Lemma \ref{idlmm}, for any $A_1\subseteq N_0$,
\begin{align*}
p(\Lambda, 0,A_1) &= \frac{1}{|\Lambda|}\sum_{A_2\subseteq \Lambda \setminus N_0}\frac{1}{{|\Lambda|-1\choose |A_1|+|A_2|}}\\
 &= \frac{1}{|\Lambda|}\sum_{k=0}^{|\Lambda\setminus N_0|}\sum_{A_2\subseteq \Lambda \setminus N_0, \, |A_2|=k}\frac{1}{{|\Lambda|-1\choose |A_1|+k}}\\
&= \frac{1}{|\Lambda|}\sum_{k=0}^{|\Lambda\setminus N_0|}\frac{{|\Lambda\setminus N_0|\choose k}}{{|\Lambda|-1\choose |A_1|+k}}\\
&= \frac{1}{|N_0|{|N_0|-1\choose |A_1|}}.
\end{align*}
This shows that when $0\in \Lambda'$, $\ee(S_0)$ depends only on $k$, $\beta$, $d$ and the random field distribution, and not on $\Lambda$. 

On the other hand, by the first inequality of Theorem \ref{normcomb}, 
\begin{align}\label{varf}
|\var(F)-\ee(S)|&\le \frac{1}{2}\sum_{i\in \Lambda} (2\ep_{2,i}m_{2,i}+ \ep_{2,i}^2),
\end{align} 
where $m_{2,i}=\|\Delta_i F\|_{L^2}$ and $\ep_{2,i} = \|\Delta_i F - g_i\|_{L^2}$, as before. Proceeding as in the proof of \eqref{deltaf}, we get $m_{2,i}\le C$ for all $i$. Similarly, proceeding as in the proofs of \eqref{ep1} and \eqref{ep2}, we get that for any $i\in \Lambda$, $\ep_{2,i}\le C$, and for $i\in \Lambda'$, 
\[
\ep_{2,i}\le C\delta_k^{1/4},
\]
where $\delta_k$ is defined as in \eqref{deltadef}. 
By \eqref{llbd} and \eqref{varf}, this gives 
\begin{equation}\label{varf2}
|\var(F)-\ee(S)|\le C\delta_k^{1/4}|\Lambda| + Ck^d |\partial \Lambda|.
\end{equation}
Now let $F_n$ and $\Lambda_n$ be as in the statement of Theorem \ref{finite}. By \eqref{ss0} and \eqref{varf2}, it follows that for each $k$, there is some number $a_k$ depending only on $k$, $\beta$, $d$ and the random field distribution, and not on the sequence $\{\Lambda_n\}_{n\ge 1}$, such that
\[
\limsup_{n\to \infty}\biggl|\frac{\var(F_n)}{|\Lambda_n|} - a_k\biggr| \le C\delta_k^{1/4}.
\]
Since $\delta_k\to 0$ as $k\to \infty$, this shows that $\{a_k\}_{k\ge 1}$ is a Cauchy sequence. Let $a$ be the limit of this sequence. Then $a$ depends only on $\beta$, $d$ and the random field distribution, and $\var(F_n)/|\Lambda_n|$ converges to $a$ as $n\to \infty$. This completes the proof of Theorem \ref{finite}, except for the last assertion about $d\le 2$. When $d\le 2$, the famous uniqueness result of \citet{aw90} for the infinite volume Gibbs state implies that 
\[
\lim_{k\to\infty} \si_{\Lambda_{i,k},+} = \lim_{k\to \infty} \si_{\Lambda_{i,k},-}. 
\]
This, together with FKG, implies that instead of \eqref{lambdadiff} we have the stronger estimate
\eq{
\ee(\sup_\gamma|\si_{\Lambda,\gamma} - \si_{\Lambda_{i,k},+}|) &\le \ee|\si_{\Lambda_{i,k}, -} - \si_{\Lambda_{i,k},+}| \to 0 \text{ as } k\to \infty.
}
The rest of the proof goes through as before.

\section{Proof of Theorem \ref{infinite}}\label{infiniteproof}
In this section $C$ will denote any positive constant that depends only on $d$ and the random field distribution. The value of $C$ may change from line to line or even within a line. As before, we will only present the proof for the plus boundary condition, since the argument for the minus boundary condition is the same.

Fix a finite nonempty set $\Lambda \subseteq \zz^d$ and consider the RFIM on $\Lambda$ with plus boundary condition. By the assumed condition on the random field distribution, the random field $\phi_i$ at a site $i\in \Lambda$ can be expressed as $u(Z_i)$, where $(Z_i)_{i\in \Lambda}$ are i.i.d.~standard normal random variables and $u$ is a Lipschitz map. Moreover, since the random field distribution is continuous, the ground state is unique with probability one. Let $\hat{\sigma}$ denote the ground state and let $G$ denote the energy of the ground state. Let $F_\beta$ denote the free energy at inverse temperature $\beta$ and let $\si_\beta$ denote the expected value of $\sigma_i$ at inverse temperature $\beta$ (on $\Lambda$, under plus boundary condition). Then it is not hard to show that 
\begin{equation*}\label{glim}
G = \lim_{\beta \to\infty} F_\beta.
\end{equation*}
Moreover, by the uniqueness of the ground state, it follows easily that almost surely,
\begin{equation}\label{silim}
\hat{\sigma}_i = \lim_{\beta \to \infty}\si_\beta.
\end{equation}
Let $\partial_i F_\beta$ be the derivative of $F_\beta$ with respect to $Z_i$. Then 
\[
\partial_i F_\beta = -u'(Z_i)\si_\beta.
\]
Thus, with probability one,
\begin{equation*}\label{glim2}
 \lim_{\beta\to \infty}\partial_i F_\beta= -u'(Z_i)\hat{\sigma}_i. 
\end{equation*}
Call the above limit $\partial_i G$. It is now easy to see from Proposition \ref{normprop} that Theorem \ref{normcont} may be applied to the function $G$, treating $\partial_i G$ as its partial derivative with respect to $Z_i$. 

For each $i\in \zz^d$, let $\hat{\sigma}_i^k$ be the ground state value of the spin at site $i$ in the RFIM on a box of side-length $2k+1$ centered at $i$ with plus boundary condition. By \eqref{silim} and the FKG property of the RFIM, it follows (similarly as in the proof of Theorem \ref{finite}) that $\hs_i^k\ge \hs_i^{k+1}$ for all $k$.  Let
\[
\hs_i^\infty := \lim_{k\to \infty} \hs_i^k.
\]
In particular, if we let
\begin{equation}\label{deltadef2}
\delta_k := \ee|\hs_i^k - \hs_i^\infty|,
\end{equation}
then 
\[
\lim_{k\to \infty} \delta_k = 0. 
\]
Now fix some $k$. If $\Lambda'$ is defined as in the proof of Theorem \ref{finite}, then for any $i\in \Lambda'$, 
\begin{equation*}\label{lambdadiff2}
\ee|\hs_i - \hs_i^k|\le \delta_k.
\end{equation*}
Let $g_i := -u'(Z_i)\hs_i^k$. Then the above inequality shows that when $i\in \Lambda'$, 
\[
\ep_{4,i} := \|\partial_i G - g_i\|_{L^4}\le C\delta_k^{1/4}.
\]
When $i\in \Lambda \setminus\Lambda'$, we trivially have $\ep_{4,i}\le C$. Also, clearly,
\[
m_{4,i} := \|\partial_i G\|_{L^4}\le C. 
\]
Let $N_i$ be as in the proof of Theorem \ref{finite}. From \cite{wehraizenman90}, we know that 
\[
\sigma^2:= \var(G)\ge C|\Lambda|.
\]
Armed with these estimates, we may now proceed as in the proof of Theorem~\ref{finite}, and using  Theorem~\ref{normcont} instead of Theorem~\ref{normcomb}, we get
\begin{align*}
\tv(\mu, \nu) &\le C \delta_k^{1/4} + Ck^{d}\frac{|\partial \Lambda|}{|\Lambda|} + \frac{Ck^{d/2}}{\sqrt{|\Lambda|}},
\end{align*}
where $\mu$ is the law of $(G-\ee(G))/\sigma$, and $\nu$ is the standard normal distribution. 

Let $G_n$ and $\Lambda_n$ be as in the statement of Theorem \ref{infinite}. Let $\mu_n$ be the law of $(G_n-\ee(G_n))/\sqrt{\var(G_n)}$. Since $|\partial \Lambda_n| = o(|\Lambda_n|)$ as $n\to\infty$, the above bound shows that 
\[
\limsup_{n\to\infty} d_{\textup{W}}(\mu_n,\nu)\le C\delta_k^{1/4}. 
\]
However, $k$ is arbitrary, and $\delta_k\to0$ as $k\to \infty$. This shows that $\mu_n$ converges to $\nu$ in the Wasserstein metric.

To complete the proof of Theorem \ref{infinite}, it only remains to show that the ratio $\var(G_n)/|\Lambda_n|$ tends to a finite nonzero limit.  As before, fix $k\ge 1$ and let $N_i$ and $g_i$ be as above. Consider $g_i$ as a function of $(Z_j)_{j\in \Lambda}$. For each $i$, let $Z_i'$ be an independent copy of $Z_i$, and for each $0\le t\le 1$, let
\[
Z_i^t := \st Z_i + \sst Z_i'.
\]
Let $g_i^t$ be the value of $g_i$ after replacing each $Z_j$ by $Z_j^t$. Then the quantity $S$ of Theorem \ref{normcont} is simply
\[
\int_0^1 \frac{1}{2\st}\sum_{i\in \Lambda} g_i g_i^tdt. 
\]
This can be rewritten as
\[
S = \sum_{i\in \Lambda} S_i,
\]
where
\begin{equation*}\label{sieq2}
S_i := \int_0^1\frac{1}{2\st}g_i g_i^t dt.
\end{equation*}
By the definitions of $g_i$, $N_i$ and $\Lambda'$, it follows that $\ee(S_i)$ is the same for all $i\in \Lambda'$. Without loss of generality, suppose that the origin $0$ is in $\Lambda'$. Thus,
\[
\ee(S) = |\Lambda'| \ee(S_0)+\sum_{i\in \Lambda\setminus \Lambda'} \ee(S_i).
\]
As in the proof of Theorem \ref{finite}, this gives
\begin{equation}\label{ss02}
|\ee(S)-|\Lambda|\ee(S_0)|\le Ck^d|\partial \Lambda|.
\end{equation}
Moreover, it is clear that when $0\in \Lambda'$, $\ee(S_0)$ depends only on $k$,  $d$ and the random field distribution, and not on $\Lambda$. On the other hand, by the first inequality of Theorem \ref{normcont}, 
\begin{align*}
|\var(G)-\ee(S)|&\le \sum_{i\in \Lambda}(2 \ep_{2,i}m_{2,i}+\ep_{2,i}^2),
\end{align*} 
where $m_{2,i}=\|\partial_i G\|_{L^2}$ and $\ep_{2,i} = \|\partial_i G - g_i\|_{L^2}$. Proceeding as in the proof of Theorem \ref{finite}, this gives
\begin{equation}\label{varf22}
|\var(G)-\ee(S)|\le C\delta_k^{1/2}|\Lambda| + Ck^d |\partial \Lambda|,
\end{equation}
where $\delta_k$ is now defined as in \eqref{deltadef2}.

Let $G_n$ and $\Lambda_n$ be as in the statement of Theorem \ref{infinite}. By \eqref{ss02} and \eqref{varf22}, it follows that for each $k$, there is some number $a_k$ depending only on $k$, $d$ and the random field distribution, and not on the sequence $\{\Lambda_n\}_{n\ge 1}$, such that
\[
\limsup_{n\to \infty}\biggl|\frac{\var(G_n)}{|\Lambda_n|} - a_k\biggr| \le C\delta_k^{1/2}.
\]
It is now easy to complete proof as in the last part of the proof of Theorem~\ref{finite}. The case $d\le 2$ also follows as before, using the uniqueness theorem of \citet{aw90} (which also holds for the ground state).

\section*{Acknowledgments}
I thank Persi Diaconis for a number of useful comments, and Nguyen Tien Dung for pointing out some omissions in the first draft. I also thank the anonymous referees for several useful suggestions.

\end{document}